\documentclass[runningheads]{llncs}

\usepackage[T1]{fontenc}
\usepackage{graphicx}
\usepackage{verbatim}
\usepackage{latexsym}
\usepackage{amssymb}
\usepackage{amsmath}
\usepackage{url}
\newcommand{\boxfigurebegincmd}{\small}

\newenvironment{pagewideframe}
  {%
   \hbox\bgroup
     \newdimen{\pagewideframewidth}%
     \setlength{\pagewideframewidth}{\hsize}%
     \addtolength{\pagewideframewidth}{-2\fboxsep}%
     \addtolength{\pagewideframewidth}{-2\fboxrule}%
     \setbox0=\hbox\bgroup
       \begin{minipage}[t]{\pagewideframewidth}%
          \setlength{\parindent}{0pt}%
          \setlength{\parskip}{4pt plus 2pt minus 1pt}%
          \ignorespaces
  }{%
          \hrule width \hsize height 0pt depth 0pt%
       \end{minipage}%
     \egroup
     \fbox{\unhbox0}%
   \egroup
  }

\newenvironment*{boxfigure}[3][htbp]%
  {%
    \begin{figure*}[#1]%
      \newcommand{\boxfigurelabel}{#2}%
      \newcommand{\boxfigurecaption}{#3}%
      \begin{pagewideframe}%
        \boxfigurebegincmd
  }{%
      \end{pagewideframe}%
      \caption{\boxfigurecaption}\label{\boxfigurelabel}%
    \end{figure*}%
  }%

\newcommand{\pagewideframehline}
  {\break
   \hbox
     {\setlength{\dimen0}{\hsize}%
      \addtolength{\dimen0}{2\fboxsep}%
      \kern-\fboxsep
      \vrule width \dimen0 height \fboxrule depth 0pt%
      \kern-\fboxsep}%
   \break}%

\newcommand{\safetext}[1]{{\ifmmode\mathchoice{\mbox{#1}}{\mbox{#1}}{\mbox{\scriptsize#1}}{\mbox{\scriptsize#1}}\else#1\fi}}
\newcommand{\upsfmdtext}[1]{\safetext{\sffamily\upshape\mdseries#1}}

\newcommand{\la}{\lambda}

\newcommand{\pbind}{\bind{\Pi}}
\newcommand{\rhoProj}[1]{\rho^{P_{..{#1}}}}

\newcommand{\odec}[2]{{#1}\mskip1.75mu{:}\mskip1.75mu{#2}}
\newcommand{\TE}{\msf{TE}}

\newcommand{\lapp}[2]{{{#1}\,{#2}}}
\newcommand{\lappThree}[3]{{{#1}\,{#2}\,{#3}}}
\newcommand{\lappFour}[4]{{{#1}\,{#2}\,{#3}\,{#4}}}
\newcommand{\lappFive}[5]{{{#1}\,{#2}\,{#3}\,{#4}\,{#5}}}
\newcommand{\mbin}{\mathbin{\in}}
\def\rhaak{\right)}
\def\lhaak{\left(}
\def\disp{\displaystyle}

\newcommand{\Dterms}{\Delta\mbox{-terms}}
\def\dom#1{\mathsf{vars}\lhaak #1\rhaak}
\newcommand{\subj}{\mathsf{var}}
\def\rrightarrow{\mathrel{\rightarrow \hspace{-.6em}\rightarrow}}
\def\mybigcapdot{\mathrel{\bigcap \hspace{-.7em}\cdot }}
\newcommand{\isnf}{\msf{isnf}}
\newcommand{\nf}{\msf{nf}}

\newcommand{\subst}[3]{{#1}[{#2}\mathbin{:=}{#3}]}
\newcommand{\ube}{\mathrel{\underline{\beta}}}
\newcommand{\fv}[1]{\mathsf{FV}({#1})}
\newcommand{\be}{\beta}
\newcommand{\n}{\upsfmdtext{f}}
\newcommand{\emptyContext}{\varepsilon}
\newcommand{\DeclarationSet}{\msf{Declaration}}
\newcommand{\RDeclarationSet}{\msf{RDeclaration}}

 \newcommand{\newdec}[3]{{{#1}\,{#2}\,{:}\,{#3}}}
 \newcommand{\rdec}[2]{{{#1}\,{#2}}}
 \newcommand{\mbto}{\mathbin{\to}}
 \newcommand{\binder}[2]{{#1}\mskip1.5mu{#2}.\,}
 \newcommand{\dsort}{\mathsf{sort}}
 \newcommand{\reqSort}{\mathsf{rsort}} 
\newcommand{\typeAsSort}{\mathsf{tsort}}
 \newcommand{\rdecname}{\mathsf{rdec}}
 \newcommand{\newapp}[2]{{{#1}\,{#2}}}
 \newcommand{\app}[3]{{{#1}\,{#2}\,{#3}}}
\newcommand{\appFive}[5]{\app{\app{#1}{#2}{#3}}{#4}{#5}}

\newcommand{\judgeKR}[6]{{#3}\vdash^{#2}_{#1}{#4}\mathbin{:}{#5}\mathbin{:}{#6}}
\newcommand{\newjudgeKR}[5]{{#3}\vdash^{#2}_{#1}{#4}\mathbin{:}{#5}}

 \DeclareMathAlphabet{\mathbfit}{OML}{cmm}{b}{it}
\def\Rules{\ensuremath{\mathbfit{R}}}
 \newcommand{\bind}[3]{{\binder{#1}{#2}{#3}}}
 \newcommand{\dottedM}{\mycrown{M}}
 \newcommand{\lbind}{\bind{\lambda}}
\newcommand{\ion}{\mathord{\overline{\in}}}
\newcommand{\oin}{\mathord{\overline{\in}}}
\newcommand{\boxedwhat}{\fbox{???}}
\newcommand{\dottedU}{\mycrown{U}}
\newcommand{\TermSet}{\msf{Term}}

\newcommand{\FreeTermSet}{\msf{LTerm}}
\newcommand{\ContextSet}{\msf{Context}}
\newcommand{\RContextSet}{\msf{RContext}}
\newcommand{\msf}{\mathsf}

\newcommand{\mycrown}[1]{\dot{#1}}
\begin{document}
\title{Intersection Types via Finite-Set
   Declarations}
\author{Fairouz Kamareddine \and Joe Wells}

\authorrunning{Fairouz Kamareddine and Joe Wells}
\institute{Heriot-Watt University, Edinburgh, UK\\
}
\maketitle              
\begin{abstract}
The $\la$-cube is a famous pure type system (PTS) cube of eight powerful
explicit type systems that include the simple, polymorphic and depen-
dent type theories. The $\la$-cube only types Strongly Normalising (SN) terms but not all of them.
It is well known that even the most powerful system of the $\la$-cube can
only type the same pure untyped $\la$-terms that are typable by the higher-order polymorphic implicitly typed $\la$-calculus $F_\omega$, and that there is an
untyped $\la$-term $\dottedU$
that is SN but is not typable in $F_\omega$ or the $\la$-cube.
Hence, neither system can type all the SN terms it expresses.
In this paper, we present the $\n$-cube, an extension of the $\la$-cube with finite-set declarations (FSDs) like $y \ion \{C_1,\cdots, C_n\} :B$ which means that $y$ is of type $B$ and
can only be one of $C_1, \cdots, C_n$. The novelty of our FSDs is that they allow to represent intersection types as $\Pi$-types. We show how
to translate and type the term $\dottedU$
in the f-cube using an encoding of intersection types based on FSDs. Notably, our translation works without
needing anything like the usual troublesome intersection-introduction
rule that proves a pure untyped $\la$-term $M$ has an intersection type $\Phi_1\cap \cdots\cap \Phi_k$ using  $k$ independent subderivations. As such, our approach
is useful for language implementers who want the power of intersection
types without the pain of the intersection-introduction rule.
\end{abstract}
\section{The Troublesome Intersection-Introduction Rule}
  Type theory was first developed by Bertrand Russell to avoid
the contradictions in Frege's work. Since, type theory was adapted and
used by Ramsey, Hilbert/Ackermann and Church and later exploded
into powerful exciting formalisms that played a substantial role in the
development of programming languages and theorem provers, and in the
verification of software. As advocated by Russell, type theory remains to
this day a powerful tool at avoiding loops/contradictions and at characterising strong normalization (SN). There are 2 styles of typing: {\em explicit}
(\`a la Church) as in $\la x : T.x$ and {\em implicit} (\`a la Curry) as in $\la x:x$. We
call the latter {\em untyped}. A type assignment engine needs to work harder
to assign a type to an untyped term than to an explicitly typed one.

Intersection types were independently invented near the end of the 1970s by Coppo and Dezani \cite{Cop+Dez-Cia:NDJFL-1980-v21n4} and
Pottinger \cite{Pottinger:HBC-1980} with aims such as the analysis of normalization properties, which requires a very precise
analysis (see also~\cite{DBLP:journals/tcs/BonoVB08,DBLP:journals/iandc/LiquoriR07,DBLP:journals/entcs/Rocca02}).
Aiming to make use of this precision, the members of the Church Project worked on a compiler that used intersection types not only to
support the usual kind of type polymorphism but also to represent a precise polyvariant flow analysis that was used to enable optimizing
representation transformations~\cite{Wel+Dim+Mul+Tur:JFP-2002-v12n3-no-note}.
Among the many challenges that were faced, a major difficulty was the intersection-introduction typing rule, which made it complicated to do
local optimizations (an essential task for a compiler) while at the same time retaining type information and the ability to verify it.
The intersection-introduction rule usually looks like this:
\[\frac{E \vdash M:\sigma \hspace{0.5in} E \vdash M:\tau}{E\vdash M:\sigma\cap\tau} \hfill \quad{(\mbox{$\cap$-Intro})}\]

The proof terms are the same for both premises and the conclusion!
No syntax is introduced.
A system with this rule does not fit into the proofs-as-terms (PAT,
a.k.a.\ propositions-as-types and Curry/Howard) correspondence,
because it has proof terms that do not encode deductions.
This trouble is related to the fact that the $\cap$ type constructor is not a truth-functional propositional connective, but rather one that
depends on the proofs of the propositions it connects sharing some specific key structural details but not all
details~\cite{Venneri:JLC-1994-v4n2}.

There is an immediate puzzle in how to make a type-annotated variant of the system.
The usual strategy fails immediately, e.g.:
\[\frac{E \vdash (\la x:\sigma.x):\sigma\rightarrow \sigma \hspace{0.5in} E \vdash  (\la x:\tau.x):\tau\rightarrow \tau }{E\vdash  (\la x:\boxedwhat.x):(\sigma\rightarrow \sigma)\cap (\tau\rightarrow\tau)}\]
Where \fbox{???} appears, what should be written?  A compiler using intersection types must have some way of organizing the type
information of the separate typing subderivations for the same program points, because a transformation at a program point must
simultaneously deal with all of the separate subderivations.  It would be nice if this was principled rather than \emph{ad hoc}.

The various solutions to this problem each have their own strengths and weaknesses.
The most basic strategy is to accept the usual style of ($\cap$-\mbox{Intro}), and not try to have proof terms that contain type
annotations \cite{Dunfield:JFP-2014-v24,Oli+Shi+Alp:ICFP-2016}.%
This is fine if the plan is to discard much or most type information early in compilation, but is unhelpful
if checkable type-correctness is to be maintained through program transformations.

Another strategy is to have proof terms whose structure makes multiple copies of
subprograms typed with the ($\cap$-\mbox{Intro}) rule~\cite{Wel+Dim+Mul+Tur:JFP-2002-v12n3-no-note}.
This means that proof terms can not be merely annotated versions of the $\la$-terms being typed, because the branching structure of the proof
terms can not be the same as that of the $\la$-terms.
Our experience is that this makes local program transformations complicated and awkward if checkable type-correctness is to be maintained.
Another option is to limit the possible typings and the set of typable terms, and accept not having the full power of intersection
types \cite{Com+Pie:MSCS-1996-v6n5}.
However, this causes difficulty when the purpose of using intersection types is to support arbitrarily precise program analyses, and also when program transformations go outside of the restricted set of typings that are supported.

These issues led us to search for ways to get the power of intersection types without the usual multiple-premise-style ($\cap$-\mbox{Intro}) rule (\cite{Wel+Haa:ESOP-2002-no-note} gives an overview of earlier solutions). Some solutions do manage to merge the premises of the ($\cap$-\mbox{Intro}) rule into just one premise, but do not provide proof terms containing type information for variable bindings or any other easy-to-manipulate representation of typing derivations \cite{Ron+Rov:CSL-2001,Cap+Lor+Ven:BOTH-2001}.

\section{Finite-Set Declarations and Encoding $\cap$-Types}

Pure type systems (PTSs) were independently given in~\cite{Berardi:PhD:1990} and~\cite{Terlouw:GSTT:1989}
and have been used to reason simultaneously about families of type systems and logics.
The well known $\lambda$-cube of 8 specific PTS's \cite{Barendregt:HLCS-1992} captured the
core essence of polymorphic, dependent and 
Calculus of Constructions (CoC) systems.  
PTSs were extended with
{\em definitions}~\cite{Blo+Kam+Ned:IAC-1996-v126n2,DBLP:conf/lfcs/SeveriP94} in order to better represent mathematics and computation.
In addition to old $x:A$ declarations, 
these definitions allow declarations of the form $x=_d D:A$ which declare $x$ to be $A$ and to have type $B$.
Of course extra typing rules are added to type
the new terms with definitions.

While analysing the troublesome ($\cap$-\mbox{Intro}) rule, we noted
that definitions can be generalised to represent intersection types. This article
is the result of this observation. We present the new syntax which extends the $\la$-cube with finite set declarations that generalise definitions to support intersection types without the troublesome ($\cap$-\mbox{Intro}) rule.
Instead of definitions  $\la x=_d D:A. B$, the new syntax adds finite-set declarations (FSDs) $\la x\ion \{D_1, D_2, \cdots, D_n\} :A. B$ where $n\geq 1$.  This latter term is a function like $\la x : A.B$ 
that also requires its argument $x$ to exhibit at most the behaviors of $D_1, \cdots, D_n$.

We show that FSDs give the power of intersection types by translating an
intersection type $\Phi_1\cap\cdots\cap\Phi_k$ to a $\Pi$-term of the form:
$$\Pi z \ion \{P_{1,k},  \cdots, P_{k,k}\} :\ast_k. z\Phi_1\cdots\Phi_k$$
  where $\ast_k$ abbreviates $\underbrace{\ast\rightarrow\cdots\rightarrow\ast\rightarrow\ast}_{k \: \tiny{\mbox{arrows}}}$ and $P_{i,k} = (\la x_1:\ast.\cdots .\la x_k:\ast.x_i)$ picks the $i$-th of $k$ arguments. So, if $z = P_{i,k}$, then $z\Phi_1\cdots\Phi_k = \Phi_i$. Notably, our translation from intersection types works without anything in the translation result
  like the usual ($\cap$-\mbox{Intro}) rule that proves a pure untyped $\la$-term $\dottedM$ has an intersection type
  $\Phi_1\cap\cdots\cap\Phi_k$ using $k$ independent subderivations. In
PTS style, these ``subderivations'' are done simultaneously because in the scope
of the declaration $z \ion \{P_{1,k},  \cdots, P_{k,k}\} :\ast_k$, the type $z(A_1\rightarrow B_1)\cdots (A_k\rightarrow B_k)$ can be converted to the equal type $(z A_1\cdots A_k) \rightarrow (z B_1\cdots B_k)$.

Existing methods for
getting rid of the ($\cap$-\mbox{Intro}) rule which support some kind of type
equivalence~\cite{Wel+Haa:ESOP-2002-no-note} differ from our method in that our type equivalence is a consequence of the FSDs restriction on $z$, which is added to a system supporting $\Pi$-types, including those usually referred to as ``higher-order polymorphic types''
and ``dependent types''. Thus, FSDs might be a good way to add the power
of intersection types to languages like Coq, Agda, and Idris.  By supporting the full power of intersection types, FSDs might also help represent results of arbitrarily precise program analyses (e.g., a polyvariant flow analysis) in language implementations that have $\Pi$-types in their internal representation.

FSDs can do more than support a translation of intersection types. FSDs can
represent a ``definition'' by a $\be$-redex where the abstraction has a FSD with only one term in the restriction. The point of using a definition like $(\la x\ion \{D\}:C.B)D$ instead of a $\be$-redex like $(\la x :C.B)D$ is that in the former, $x$ is $D$ can be used to
help justify type-correctness of $B$, which might otherwise require replacing many
instances of $x$ by $D$. Although using an FSD instead of a traditional definition
requires forming the abstraction $\la x\ion\{D\}:C.B$ and its type, and hence an additional type formation rule, it is worth noting that, adding support for FSDs to
suitable systems does not require huge changes. This is the case since although
not always prominently stated in theory papers, in practice, proof assistants
support definitions in the formal systems of their implementations.

Section~\ref{sec3} introduces into the $\la$-cube the new feature of finite-set declarations. Section~\ref{sec4} presents
examples that demonstrate how the new syntax can be used to simulate intersection types and shows how a term of Urzyczyn which is untypable in the $\la$-cube
can be typed in the \n-cube.

Let $\mathbb{N}_1 = \mathbb{N}\setminus\{0\}$ be the positive natural numbers. Given $i, j \in \mathbb{Z}$,
define $i .. j = \{k\in \mathbb{Z} \mid i\leq k\leq j\}$ and $[i .. j) = i .. (j-1)$.   Write $|S|$
  for the size of set $S$.
  
As usual, the composition $X\circ Y$ is $\{(y,x) \mid \exists z. (y, z) \in Y \mbox{ and } (z, x) \in X \}$. Given
$n \in \mathbb{N}_1$, let $X^n$ be such that $X^1 = X$ and $X^{i+1} = X \circ (X^i)$ for $i \in \mathbb{N}_1$.

\section{Extending the Syntax of the $\la$-cube}
\label{sec3}
The new syntax has declarations of the form $x\rho:A$, meaning that the variable
$x$ has type $A$ and $x$ also obeys restriction $\rho= \ion\{A_1,\cdots, A_k\}$ for $k\in \mathbb{N}$. When
$k = 0$, then $x\ion\{\} : A$
is the usual (unrestricted)
declaration $x:A$ of the $\la$-cube. When $k \in \mathbb{N}_1$, we get the new restricted finite-set declaration (FSD) $x \rho :B$ which only permits $x$ to be one of the $A_i$'s. These FSDs are the innovation of the \n-cube.

\begin{definition}[Syntax]
  \label{def:synt}
  Figures~\ref{fig:syntax} and~\ref{fig:syntaxintertype} use the usual pseudo-grammar notation to define the sets of syntactic entities. Many of these  entities  are clear from the usual type-free/typed $\la$-calculus. In the $\la$-cube we have the sorts $\ast$ and $\Box$, the rules
$\Rules\in {\msf{RuleSet}}$, and the declarations $\delta$ and contexts $\Delta$.  In the extended cube
however, the declarations are a generalised version of those of the $\la$-cube and
a declaration $\delta$ may not only be of the the usual form $x : A$ (which we also
write as $x \diamond : A$ and states that $x$ is of type $A$) but may also be of the form
$x\ion \{A_1, \cdots,A_i\}:A$ which states that $x$ is declared as any of the $A_j$s (for $j\in 1..i$) and is of type $A$.  We call declarations of the form $x\ion\{A_1, \cdots, A_i\}:A$, where $i\not = 0$ restricted declarations and these belong to $\RDeclarationSet$.

Expressions of the $\n$-cube
  and the pure $\la$-calculus  are given respectively by $\TermSet$ and $\FreeTermSet$.
  Figure~\ref{fig:syntax} defines default set ranges for metavariables.
 $\msf{Variable}$ is used as the set of what we call names. There are two name classes:
$\msf{Variable}^\ast$ marked with sort $\ast$ and $\msf{Variable}^\Box$ marked with sort $\Box$. For embedding
  reasons, it is usual to take the type free $\la$-calculus variables $\msf{LVariable}$ to be  $\msf{Variable}^\ast$. So, $x^\ast$ and $\mycrown{x}$ range over  $\msf{Variable}^\ast = \msf{LVariable}$ whereas $x^\Box$
ranges over $\msf{Variable}^\Box$. If no confusion occurs, we simply write $x$ to range
over $\msf{Variable} = \msf{Variable}^\ast\cup \msf{Variable}^\Box$.
\vspace{-0.1in}
\begin{boxfigure}{fig:syntax}{Metavariable declarations and type free terms.}
  \noindent
  $$
  \begin{array}{@{}r@{\ \in\ }l@{\ \mathrel{::=}\ }l@{}}
      a,b,\quad i,\ldots,n,\quad q,r                        & \multicolumn{1}{@{}l@{}}{\mathbb{N}}
      \\ \varsigma                            & \msf{Sort}          & \ast \mid \Box
\\ a^\varsigma,\ldots, z^\varsigma                        & \multicolumn{1}{@{}l@{}}{\msf{Variable}^{\varsigma}}
          \\ f,g,h,\quad q,\ldots,z            & \multicolumn{1}{@{}l@{\ =\ }}{\msf{Variable}}
                                                          & \msf{Variable}^{\ast}\cup \msf{Variable}^{\Box}
    \\ \mycrown{a}, \cdots, \mycrown h,\quad\mycrown{k}, \cdots, \mycrown{z}    & \multicolumn{1}{@{}l@{\ =\ }}{\msf{LVariable}}
                                                          & \msf{Variable}^\ast
    \\ \mycrown{A},\cdots,\mycrown{Z}     & \FreeTermSet        & \mycrown{x} \mid \lambda \mycrown{x}. \mycrown{M} \mid \mycrown{M}\mycrown{N}\quad\hfill\mbox{\textit{(type-free $\la$terms)}}
     \vspace{-0.1in}
    \end{array}
  $$
\end{boxfigure}
Let $\msf{Syntax}$ be the uniton of all the sets of syntactic entities that we define in Figures~\ref{fig:syntax} and~\ref{fig:syntaxintertype}.  
Let $\chi$ range over $\msf{Syntax}$. We assume the usual assumptions of binding, $\alpha$-conversion and the Barendregt variable renaming. We take $\alpha$-convertible expressions to denote the same syntactic entities, e.g., even if $\mycrown{x} \not = \mycrown{y}$, it nonetheless holds that $\la \mycrown{x}. \mycrown{x} = \la  \mycrown{y}. \mycrown{y}$.  As usual, let $\fv{{\chi}}$ be  the collection of all variables in ${\chi}$.  We say $\chi$ is closed iff $\fv{{\chi}}=\{\}$. We assume the usual substitution in the $\la$calculus where the capture of free variables must be avoided. 
\end{definition}
\vspace{-0.1in}
\begin{boxfigure}{fig:syntaxintertype}{$\la$- and $\n$-cube systems syntax definitions.}
  \noindent
  $$
    \begin{array}{@{}r@{\: \in \:}l@{\ \mathrel{::=}\ }l@{}}
        \pi                            & \msf{Binder}        & \lambda \mid \Pi
         \\ \rho                            & \msf{Restriction}   & \overline{\in} \{A_1, \dots,A_i\} \quad \mbox{where }i\in\mathbb{N} \quad
    \\ \delta                            & \DeclarationSet     & \newdec{x}{\rho}{A}
    \\ \Delta                            & \ContextSet         & \emptyContext \mid {\delta}_1,\ldots,{\delta}_i \quad \mbox{where }i\in\mathbb{N}_1
    \\ \gamma                            & \RDeclarationSet    & \rdec{x}{\rho} \quad\mbox{where $\rho \not = \diamond$}
    \\ \Gamma                            & \RContextSet        & \emptyContext \mid {\gamma}_1,\ldots,{\gamma}_i \quad \mbox{where }i\in\mathbb{N}_1
    \\ A,\ldots,H,\quad J,\quad L,\ldots,W       & \TermSet            & \varsigma \mid x \mid \bind{\pi}{\delta}{A} \mid \newapp{A}{B}
       \\\multicolumn{1}{@{}l@{}}{}
                                    & \msf{Rule}          & (\varsigma,\varsigma') 
       \\ \Rules                       & \multicolumn{1}{@{}l@{\ ::=\ }}{\msf{RuleSet}} & \{X\subseteq\msf{Rule} \mid (\ast,\ast)\in X\}
            \vspace{-0.1in}
    \end{array}
    $$
    \begin{itemize}
    \item
            Define the null restriction $\diamond= \ion\{\}$ and write
            $x\diamond:A$ as the usual declaration $x:A$.
               \item
      Define the restriction declarations of $\Delta$ by:  $\rdecname(\newdec{x}{\diamond}{A},\Delta) = \rdecname(\Delta);$\\$\rdecname(\varepsilon) = \varepsilon;\quad$ and 
      $\rdecname(\newdec{x}{\oin\{A_1,\dots,A_k\}}{A},\Delta) =  \rdec{x}{\oin\{A_1,\dots,A_k\}},\rdecname(\Delta)$.
       \item
            For the $\la$-cube, $\msf{Restriction} = \{\diamond\}$; $ \RDeclarationSet = \emptyset$ and $\rdecname(\Delta) = \varepsilon$.

    \item
      Define the variables of a $\delta$ or $\gamma$ by: $\dom{x\rho:A} =\dom{x\rho} = \{x\}$ and define
      $\dom{\varepsilon}=\{\}$;
  $\hspace{.2in}$ $\dom{\delta,\Delta}={\subj(\delta)}\cup\dom{\Delta}$;$\quad$
  $\dom{\gamma,\Gamma}={\subj(\gamma)}\cup\dom{\Gamma}$.
    \item
      Define $\sharp(B)$, the degree of $B$ where  $\sharp \in \msf{Term}\mbto{0..3}$ by:\\ $\sharp(\Box) = 3$, $\sharp(\ast) = 2$, $\sharp(x^\varsigma) =\sharp(\varsigma)-2$, and
      $\sharp(\bind\pi\delta A) = \sharp(\newapp AB) = \sharp (A)$.
    \item
      Define $\reqSort$ and $\typeAsSort$       by:
      \begin{itemize}
        \item
          If $\sharp(B) = 0$ then $\reqSort(B) = \ast$ and if $\sharp(B) = 1$ then $\reqSort(B) = \Box$.
          \item
            If $\sharp(B) = 1$ then $\typeAsSort(B) = \ast$ and if $\sharp(B) = 2$ then $\typeAsSort(B) = \Box$.
             \end{itemize}
    \end{itemize}
    \vspace{-0.1in}
\end{boxfigure}
\vspace{-0.1in}
\begin{definition}[Rewriting]
  \label{def:rewriting}
  We use the usual notion of \emph{compatibility} of a relation on syntactic entities.
   Let $\ube$ and $\beta$ be the smallest compatible relations where:
  $$
    \begin{array}{@{}r@{\ \;}c@{\ \;}l@{}}
       (\lbind{x}{\mycrown{M}})\mycrown{N}           & \ube         & \mycrown{M}[x:=\mycrown{N}]
    \\[0.5ex]
      \newapp{(\bind{\la}{\newdec{x}{\rho}{A}}{B})}{C} & \mathrel{\beta}  & B[x:=C]
    \end{array}
  $$
    For $r\in\{\ube,\be\}$, let $\rightarrow_r$, $\rrightarrow_r$, and $=_r$ be defined as usual:
    $(\rightarrow_r)=(r)$; and
    $\rrightarrow_r$ is the smallest transitive relation containing $r$ that is reflexive on $\msf{Syntax}$; 
   and
    $=_r$ is the smallest transitive symmetric relation containing $\rrightarrow_r$.
 \end{definition}
The following theorem  shows that our rewriting rules are confluent. 
\begin{theorem}[Confluence for $\beta$]
  \label{thm:confluence}
  If $\chi_1\rrightarrow_{\be}\chi_2$ and $\chi_1\rrightarrow_{\be}\chi_3$ then there exists $\chi_4\in\msf{Syntax}$ such that
  $\chi_2\rrightarrow_{\be}\chi_4$ and $\chi_3\rrightarrow_{\be}\chi_4$.
 \end{theorem}
 \begin{proof}
\label{prf:thm:confluence}
  First, translate $\msf{Syntax}$ and $\beta$ into a higher-order rewriting (HOR) framework, e.g.,  van Oostrom's
  framework~\cite{vanOostrom:CAHOR-1994}.
  It is then straightforward to show that $\beta$ is \emph{orthogonal}.
  It follows by a standard HOR result that $\beta$ is \emph{confluent}.
  \hfill $\boxtimes$
\end{proof}

\begin{definition}[Normal Forms]
  \label{def:nf-sn}
  A syntactic entity $\chi$ \emph{is a normal form}, written $\isnf(\chi)$, iff there is no $\chi'$ such that $\chi\rightarrow_r\chi'$ for
  $r\in\{\ube,\be\}$.
  The \emph{normal form of} $\chi$, written $\nf(\chi)$, is the unique syntactic entity $\chi'$ such that $\chi\rrightarrow_r\chi'$ for
  $r\in\{\ube,\be\}$ and $\isnf(\chi')$.
  (Note that $\nf(\chi)$ might be undefined, e.g., consider $\chi=\newapp{B}{B}\in\TermSet$ where
  $B=\lbind{\newdec{x}{\diamond}{y}}{\newapp{x}{x}}$, has no normal form (and is also not type-correct).)
    A syntactic entity $\chi$ is strongly normalizing, written $\msf{SN}(\chi)$, iff there is no infinite $r$-rewriting sequence starting at
  $\chi$ for $r\in\{\ube,\be\}$.
\end{definition}

Each of the $\la$-cube and the $\n$-cube has  8 type systems each defined by a set $\Rules$ which contains type formation rules which the $(\Pi)$ and $(\la)$ rules use to regulate the allowed abstractions. Figure~\ref{fig_par_bcube} gives the 8 systems defined by these $\Rules$s. E.g., $\widehat{\lambda C}$ uses all combinations $(\varsigma,\varsigma')$ where $\varsigma\in\{\ast,\Box\}$.
{\footnotesize
\begin{figure}
\begin{center}
\begin{tabular}{|l|l|c|c|c|c|}
\hline
$\widehat{{\lambda}{\rightarrow}}$
                 &$(\ast,\ast)$\hspace{1in}&&&\\
$\widehat{\lambda 2}$
              &$(\ast,\ast)$ &$(\Box,\ast)$\hspace{0.5in}&&\\
$\widehat{\lambda P}$&$(\ast,\ast)$&&$(\ast,\Box)$\hspace{0.5in}&\\
$\widehat{\lambda P2}$&$(\ast,\ast)$&$(\Box,\ast)$&$(\ast,\Box)$&\\
$\widehat{\lambda\underline{\omega}}$&$(\ast,\ast)$& & & $(\Box,\Box)$\\
$\widehat{\lambda\omega}$&$(\ast,\ast)$&$(\Box,\ast)$&&$(\Box,\Box)$\\
$\widehat{\lambda\mbox{P}}\underline{\omega}$   &$(\ast,\ast)$&&$(\ast,\Box)$&
$   (\Box,\Box)$\\
$\widehat{\lambda C}$ &$(\ast,\ast)$&$(\Box,\ast)$&$(\ast,\Box)$&
    $(\Box,\Box)$\\
\hline
\end{tabular}
\setlength{\unitlength}{.15mm}
\begin{picture}(210,80)
\thicklines
\put(0,-40){\line(1,0){120}}
\put(120,-40){\line(0,1){120}}
\put(120,80){\line(-1,0){120}}
\put(0,80){\line(0,-1){120}}
\put(120,-40){\line(3,1){90}}
\put(120,80){\line(3,1){90}}
\put(0,80){\line(3,1){90}}
\put(90,110){\line(1,0){120}}
\put(210,110){\line(0,-1){120}}
\put(90,-10){\line(0,1){85}}
\put(90,85){\line(0,1){25}}
\put(90,-10){\line(1,0){25}}
\put(125,-10){\line(1,0){85}}
\put(0,-40){\line(3,1){90}}
\put(1,-29){\makebox(0,0)[bl]{$\widehat{\lambda{\rightarrow}}$}}
\put(119,-39){\makebox(0,0)[br]{$\widehat{\lambda P}$}}
\put(1,109){\makebox(0,0)[tl]{$\widehat{\lambda 2}$}}
\put(121,79){\makebox(0,0)[tl]{$\widehat{\lambda P2}$}}
\put(85,-11){\makebox(0,0)[br]{$\widehat{\lambda\underline{\omega}}$}}
\put(207,-7){\makebox(0,0)[br]{$\widehat{\lambda P}\underline{\omega}$}}
\put(238,129){\makebox(0,0)[tr]{$\widehat{\lambda C}$}}
\put(91,136){\makebox(0,0)[tl]{$\widehat{\lambda\omega}$}}
\put(0,-40){\circle*{3}}
\put(120,-40){\circle*{3}}
\put(120,80){\circle*{3}}
\put(0,80){\circle*{3}}
\put(90,-10){\circle*{3}}
\put(210,-10){\circle*{3}}
\put(210,110){\circle*{3}}
\put(90,110){\circle*{3}}
\end{picture}
\hspace*{1mm}
\begin{picture}(100,60)
\put(-40,-100){\vector(1,0){60}}
\put(-40,-100){\vector(0,1){60}}
\put(-40,-100){\vector(3,1){44}}
\put(22,-100){\makebox(0,0)[l]{$(\ast,\Box)$}}
\put(-4,-65){\makebox(0,0)[l]{$(\Box,\Box)$}}
\put(-38,-40){\makebox(0,0)[l]{$(\Box,\ast)$}}
\end{picture}
\end{center}
\vspace{-0.1in}
\caption{The rule sets for the $\la$-cube and $\n$-cube}
\label{mytable}
\label{fig_par_bcube}
\end{figure}
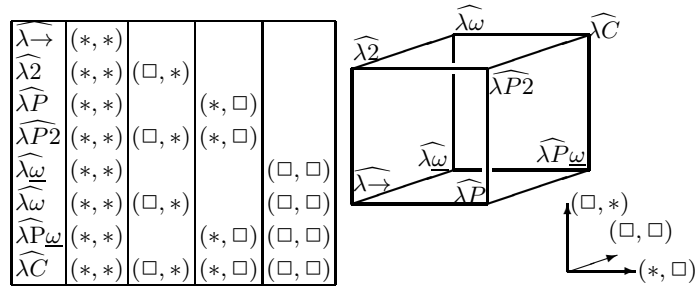
}

\vspace{-0.1in}
\begin{boxfigure}{fig:restrsat}{Restriction Satisfaction/Utilisation Judgements ($\Gamma \Vdash B\rho$)}
  \noindent
  $$ \mbox{(ref)} \ %
  \frac{i\in 1..n; \hspace{0.5in} B=_\be A_i}{\varepsilon\Vdash B\ion \{A_1, \cdots, A_n\}}
  $$
  
  $$ \mbox{(ctR)}  \ %
        \frac
          {\disp
             x \not \in \dom{\Gamma}
           ;\quad
             \binder\forall{i\mbin 1..n}\Gamma[x:=A_i] \Vdash  B[x:=A_i]\,\rho[x:=A_i]}
          {\disp \rdec{x}{\ion\{A_1, \dots, A_n\}}, \Gamma \Vdash  B\,\rho}
          $$
          \vspace{-0.1in}
  \end{boxfigure}

\begin{definition}[Typing Rules and Judgements and Type Systems]%
  \label{defsystems}
  The typing rules for the $\la$- and $\n$-cubes are given in figure~\ref{newsimplifiedtypes}.  If $\newjudgeKR{}{\Rules}{\Delta}{A}{B}$ then:
  \begin{itemize}
  \item
     In type system $\lambda^{\Rules}$ and in context $\Delta$ the term $A$ has \emph{type} $B$.
  \item
       $\Delta$, $A$, and $B$ are \emph{$\lambda^{\Rules}$-legal} (or simply legal) and $A$ and $B$ are \emph{$\Dterms$}.
  \item
    $A$ has sort $\varsigma$ if also $\newjudgeKR{}{\Rules}{\Delta}{B}{\varsigma}$ holds
    (in this case, note that  $\dsort(A)=\varsigma$).
  \end{itemize}
  Let $\judgeKR{}{\Rules}{\Delta}{A}{B}{C}$ stand for $\newjudgeKR{}{\Rules}{\Delta}{A}{B}$ and
  $\newjudgeKR{}{\Rules}{\Delta}{B}{C}$. If  $\Rules$ is omitted from $\vdash^{\Rules}$, then the reader should infer it.
  
  As we see in  Figure~\ref{newsimplifiedtypes}, the $\la$- and $\n$-cubes only differ in rules \mbox{(start)}, \mbox{(app)} and \mbox{(conv)}.  For the $\n$-cube, \mbox{(start)} is an obvious generalisation of that of the $\la$-cube (checking that the $B_j$s have the correct type), whereas
  \mbox{(app)} and \mbox{(conv)} use Figure~\ref{fig:restrsat} to check that only well behaved restrictions are used. Here, $\rdecname(\Delta)\Vdash  A\,\rho$ and $\rdecname(\Delta)\Vdash  B\ion\{C\}$ ensure that the restrictions via FSDs are activated according to Figure~\ref{fig:restrsat} so that if $\rho=\ion\{C_1, \cdots,C_n\}$, then for all $i$,  $A=_\be C_i$ modulo substitutions based on FSDs in $\rdecname(\Delta)$.  Thus, if there are no FSDs, i.e.\ $\rdecname(\Delta)=\varepsilon$, then the  $\rdecname(\Delta)\Vdash  B\ion\{C\}$ of \mbox{(conv)} becomes the $B=_\be C$ in the $\la$-cube.\footnote{The \mbox{(weak)} rule differs slighly  from that of the $\la$-cube, but a simple check shows that this formalisation of the $\la$-cube is equivalent to that of~\cite{Barendregt:HLCS-1992}.}
\end{definition}
\vspace{-0.1in}
 \begin{boxfigure}{newsimplifiedtypes}{Typing rules of the $\la$- and $\n$-cubes.}%
  $$
    \begin{array}{@{}c@{}}
        \mbox{(axiom)}
      \ %
        \emptyContext\vdash^{\Rules} \ast:\Box
    \quad\qquad
        \mbox{(weak)}
      \ %
          \disp
          \frac{\disp
                  \Delta,\delta\vdash^{\Rules} A:B
                ;\qquad
                  \Delta\vdash^{\Rules} C : D}
               {\disp\Delta,\delta\vdash^{\Rules} C : D}
               \\[3ex]
                   \mbox{(start)}
      \ %
          \disp
          \frac
            {
              \disp
              \left(
              \begin{array}{@{}l@{}}
                                 x^\varsigma\not\in\dom{\Delta}
                ;\qquad
                  \Delta\vdash^{\Rules} A : \varsigma 
                ;
              \\ \mbox{if $\rho=\ion\{B_1,\cdots,B_n\}$ then }
                \binder\forall{j\in 1..n}
                          \Delta\vdash^{\Rules} B_j: A
                  \end{array}
                        \right)                 }
            {\disp\Delta,\newdec{x^\varsigma}{\rho}{A} \vdash^{\Rules} x^\varsigma:A}
            \\[3ex]
                     \mbox{($\Pi$)}
      \ %
        \disp
        \frac
          {\disp
             \Delta,\newdec x\rho A\vdash^{\Rules}B:\varsigma
           ;\qquad
             \Delta\vdash^{\Rules}A:\varsigma'
           ;\qquad
             (\varsigma',\varsigma)\in\Rules}
          {\disp
           \Delta\vdash^{\Rules}\bind{\Pi}{\newdec x\rho A}{B}:\varsigma}
          \\[3ex]
            \mbox{($\la$)}
      \ %
        \disp
        \frac{\disp\Delta,\delta\vdash^{\Rules} B':B
         ;\qquad
        \Delta \vdash^{\Rules} \bind{\Pi}{\delta}{B}:\varsigma
             }
             {\disp\Delta\vdash^{\Rules} \bind{\la}{\delta}{B'}:\bind{\Pi}{\delta}{B}}
         \\[3ex]
           \mbox{(app)}
      \ %
          \disp
          \frac
            {\disp
               \Delta\vdash^{\Rules} F:(\bind{\Pi}{\newdec{x}{\rho}{C}}{B})
             ;\qquad
               \Delta\vdash^{\Rules} A:C
             ;\qquad
               (
                  \mbox{ if } \rho\not =\diamond
                                      \mbox{ then }
                   \rdecname(\Delta)\Vdash  A\,\rho
               )
            }
            {\disp\Delta\vdash^{\Rules} \newapp{F}{A}:B[x:=A]}
    \\[3ex]
       \mbox{(conv)}
      \ %
          \disp
          \frac{\disp
                  \Delta\vdash^{\Rules}  A:B
                ;\qquad
                  \Delta\vdash^{\Rules} C:\varsigma
                ;\qquad
                  \rdecname(\Delta) \Vdash B\ion\{C\}}
               {\disp\Delta\vdash^{\Rules} A:C}
               \vspace{-0.1in}
   
    \end{array}
    $$
    \begin{itemize}
    \item
      In the $\la$-cube, the if-then statements of \mbox{(start)} and \mbox{(app)} do not apply, and the $\rdecname(\Delta) \Vdash B\ion\{C\}$ of \mbox{(conv)} becomes $b=_\be C$ by Figure~\ref{fig:restrsat}.
    \item
      We write $\Delta\vdash^{\Rules}_{\la} A:B$ resp.\ $\Delta\vdash^{\Rules}_{\n} A:B$ for type derivation in the $\la$- resp.\ $\n$-cube.
    \end{itemize}
    \vspace{-0.1in}
\end{boxfigure}

 The next definition gives the function $\TE$  which erases types and all information at degree 1 or more from elements of $\TermSet$ to return pure untyped $\la$-terms.
  $\TE$ is like the function E of~\cite{GHRDR1993}, but is simpler because we only need $\TE(A)$
  to be meaningful when $\Delta \vdash_\n A:B:\ast$.
       If $A$ is not legal or the sort (type of the type) of $A$ is not $\ast$, then we do not care whether $\TE(A)$ is defined or if so what it
    is.

\begin{definition}[Type Erasure]
  \label{def:te}
  Let $\TE\in \TermSet\rightarrow\FreeTermSet$ be the smallest function where:
       $ \TE(x)                                   = x$ and \\
  $
    \begin{array}{@{}l@{\ =\ }l@{\qquad\qquad}l@{\ =\ }l@{}}
        \TE(\newapp{A}{B})                       & \lapp{\TE(A)}{\TE(B)} \quad\mbox{if $\sharp(B)=0$}
      &
       \TE(\lbind{\newdec{x^\ast}{\rho}{A}}{B}) & \lbind{x^\ast}{\TE(B)}
        \\
        \TE(\newapp{A}{B})                       & \TE(A)                \quad \quad \quad \quad\mbox{if $\sharp(B)=1$}
      &
   \TE(\lbind{\newdec{x^\Box}{\rho}{A}}{B}) &  \quad \quad\TE(B)
    \end{array}
    $
    \end{definition}
      Definition~\ref{defsystems} stated how we can type explicily typed terms (those of $\TermSet$).  The next definition states how to type pure type-free terms (those of 
$\FreeTermSet$).
\begin{definition}[Typability of Pure $\la$-Terms]
  \label{typabilitydef}
  A pure $\la$-term $\mycrown{M}$ is \emph{typable} iff there exist $\Delta$, $A$, and $B$ such that $\Delta \vdash_{\n} A:B:\ast$ and $\TE(A) = \mycrown{M}$.
\end{definition}

\section{Implementing Intersection Types}
    \label{sec4}
This section defines intersection types using FSDs and shows that Urzyczyn's famous term is typable in the $\n$-cube.  Throughout,  assume we are using one of the 2 systems that allow forming functions ``from types to types'' (``type constructors'') and ``from types to terms'' (``type polymorphism'').  Thus, prefix every statement with ``If $\{(\Box,\Box), (\Box,\ast)\}\subseteq \Rules$, then  $\cdots$'' and read every  ``$\vdash$'' as ``$\vdash^\Rules_\n$''.

The next definitions give pieces of syntax and abbreviations that are needed to define intersection types in the $\n$-cubes. 
\begin{definition}[General Syntax Abbreviations]
  \vspace{-0.1in}
     \begin{itemize}
        \item
    $A \mbto B
        = \bind{\Pi}{\odec{w^{\varsigma}}{A}}{B}\mbox{ where }w^\varsigma\not\in\fv{B}\mbox{ and }\varsigma=\typeAsSort(A)$.\\
         $\mbto$ associates to the right, that is, $A\mbto B\mbto C$ stands for $A\mbto(B\mbto C)$.
\item
$ \ast_0 = \ast
     \quad$ and 
       $\quad \ast_{i+1} = \ast \mbto \ast_i$.
     \item
  Given names for declarations $\delta^{v_1}$, $\ldots$, $\delta^{v_n}$, define:
  $
    \Delta^{v_1,\ldots,v_n} = \delta^{v_1},\ldots,\delta^{v_n}
     $.
     \end{itemize}
\end{definition}
\begin{definition}[Pieces of Syntax Used in  Intersection Types]
\label{defpieces}
  The translations
  use the $y_j$'s to represent type variables, use the $x_i$'s to build
  type-choice combinators, and use the $z_q^j$'s as
  variables restricted to range over type-choice combinators.
  We implement these specialized purposes by using the declarations, restrictions, and terms defined for $i,j\in\mathbb{N}$ and $q\in\mathbb{N}_1$ by:\\
  $
    \begin{array}[b]{@{}l@{\ =\ \;}l@{\quad}l@{}}
        \delta^{y_j}
      & \newdec{y_j^{\Box}}{\diamond}{\ast}
      & 
    \\
        \delta^{x_i}
      & \newdec{x_i^{\Box}}{\diamond}{\ast}
      & 
    \\
            P_{i,q}
      & 
        \binder{\lambda}{\delta^{x_1}}
          \cdots
            .\lbind{\delta^{x_q}}{x_i}
      & \mbox{where }i\in1..q
        \quad\hfill\mbox{\textit{($P$ for ``projection'')}}
    \\
        \rhoProj{q}
      & \ion\{P_{1,q},\dots,P_{q,q}\}
      & 
    \\
        \delta^{z^j_q}
      & \newdec{{z^j_q}^{\Box}}{\rhoProj{q}}{\ast_q}
        &
            \\
        \gamma^{z^j_q}
      & \rdecname(\delta^{z^j_q})
      \quad=\quad
        \rdec{z^j_q}{\rhoProj{q}}
      & 
            \end{array}
     $
\end{definition}
We are ready to define intersection types in the $\n$-cube. We only take the intersection of terms whose degree is 1 and hence whose required sort is $\Box$ (these terms  correspond to types). From definition~\ref{defpieces}, $ \rdec{z^j_q}{\rhoProj{q}}$  is  $ \rdec{z^j_q}{\ion\{P_{1,q},\dots,P_{q,q}\}}$ and hence whichever $P_{i,q}$ is chosen from $\ion\{P_{1,q},\dots,P_{q,q}\}$ for ${z^j_q}$, we get $({z^j_q}\,A_1\,\cdots\,A_q) =_\be A_i$  and we see that $ \pbind{\delta^{z^j_q}}{({z^j_q}\,A_1\,\cdots\,A_q)}$ is the intersection of $A_1,\cdots ,A_q$.
      
\begin{definition}[Intersection Types in the \n-cube]
  Given 
  $\Delta$-terms $A_1$, $\ldots$, $A_q$ where $q\in\mathbb{N}_1$ and $\sharp(A_1)=\dots=\sharp(A_q)=1$, the intersection of $A_1$, $\ldots$,
  $A_q$ is defined as:\\
  $
  \phantom\quad\phantom\quad{\mybigcapdot}\:\{A_1,\ldots,A_q\}
    \quad=\quad
      \pbind{\delta^{z^j_q}}{({z^j_q}\,A_1\,\cdots\,A_q)}
      \quad\mbox{where }{z^j_q}\notin\fv{A_i}\mbox{ for }i\in1..q.
      $
 \end{definition}
\begin{definition}[Syntax Abbreviations for Intersection Types]
  \vspace{-0.1in}
  \begin{itemize}
  \item
    For translations that use only $y_0$ from the $y$'s, let:
 $        y
      = y_0
    $ and $
        \delta^y
        = \delta^{y_0}$.
        \item
  For examples that use only one of the $z$'s at a particular arity $q\in\mathbb{N}_1$, let:
  $        z_q
      = z^0_q
    \qquad\qquad
        \delta^{z_q}
      = \delta^{z^0_q}
    \qquad\qquad
        \gamma^{z_q}
      = \gamma^{z^0_q}.
      $
    \item Let
      $\overline{A} = A \mbto A$ $\hspace{0.3in}$ $ \widetilde{A} = \overline{A} \mbto A$
      $\hspace{0.3in}$
      $\underline{A} = (A \mbto\widetilde{A}) = A \mbto((A \mbto A) \mbto A)$\\
      $\overline{A}^0 = A \mbox{ and } \overline{A}^{i+1} = \overline{\overline{A}^i}$  $\hspace{0.5in}$ $\underline{A}_0 = A \mbox{ and } \underline{A}_{i+1} = \underline{\underline{A}_i}$.
      \end{itemize}
\end{definition}

The following lemma 
sets a type-building toolkit to be used in our examples. Its proof is straightforward using the machinery of PTSs.
\begin{lemma}[Type-Building Toolkit for Examples]
  \label{exampleslem}
  The following hold:
  \vspace{-0.1in}
  \noindent
  \begin{enumerate}%
  \item
    \label{examtype5}
    \label{lem:typing-P}
     For $i,j \in \mathbb{N}$ and $q\in \mathbb{N}_1$,      $\emptyContext  \vdash \ast_i:\Box$;     
     $\emptyContext \vdash P_{i,q} : \ast_q$ for $i\in 1..q$;
           $ \delta^{z_q^j} \vdash z_q^j :\ast_q$.    
  \item
    \label{examtypestart}
        If $\Delta$ is legal and $\newdec{u^{\varsigma}}{\rho}{A}\in\Delta$, then
    $\Delta \vdash u^\varsigma: A : \varsigma$.
  \item
    \label{examtypeappn}
    Let $j \in \mathbb{N}$,
    let $i \in \mathbb{N}_1$,
    let $n\in 1..i$,
    and let $u_1, \ldots, u_n \in \{y_k\mid k\in\mathbb{N}\}$.
    Suppose for $l\in1..n$ that $A_l\in\{u_l,\overline{u_l}\}$.
    Let $\Delta$ be legal such that  $\delta^{z_i^j}\in\Delta$ and  for all $l\in1..n$,  $\delta^{u_l}\in\Delta$.
        Then
    $\Delta \vdash z_i^j A_1 \ldots A_n :\ast_{i-n}:\Box$.
  \item
    \label{examtypelast}
    If $q \in \mathbb{N}_1$, and $\Delta \vdash A:\ast_q $, and $\Delta \vdash B_i:\ast$ for $i\in[0..q)$, then for
    $j\in[0..q)$ it holds that $\Delta \vdash \newapp{A}{B_0} \cdots\,B_j : \ast_{q-(j+1)}$.
  \item
    \label{examtypelast3p}
    If  $\Delta \vdash A:\ast$ and  $\Delta \vdash B:\ast$ then $\Delta \vdash A \mbto B:\ast$.
    \item
    \label{examtypelast2}
    If  $\Delta \vdash A:\ast$ and $i \in \mathbb{N}$ then $\Delta \vdash \overline{A}^i:\ast$ and  $\Delta \vdash \underline{A}_i:\ast$
    and $\Delta \vdash
    \widetilde{A}:\ast$.
  \end{enumerate}
\end{lemma}

\subsection{Simple Examples}
\label{sec:simple-examples}

\begin{definition}[Needed Terms and Declarations] Define the following:
  $
    \begin{array}[b]{@{}l@{\;{=}\;}l@{\qquad}l@{\;{=}\;}l@{\qquad}l@{\;{=}\;}l@{}}
        U        & \lappThree{z_2}{y}{\overline{y}}
      & \delta^u & \odec{u}{U} \mbox{ where $u\not = z_2$ and  $u\not = y$}
      & V        & \lappThree{z_2}{\,\overline{y}}{\,\overline{y}^2}
    \\
        W'       & (\lbind{\delta^{z_2}}{\lbind{\delta^u}{u}})
      & W        & (\pbind{\delta^{z_2}}{\pbind{\delta^u}{U}})
                   = (\pbind{\delta^{z_2}}{(U\mbto U)})
      & \delta^w & \odec{w}{W}
    \end{array}
      $%
  \end{definition}

\begin{example}[Derivation Simulating Polymorphic Identity with Intersection Types]
  \label{derexamplely}
  All the  derivations in this example follow from Lemma~\ref{exampleslem} and the typing rules.
  \vspace{-0.1in}
  \begin{itemize}
    \item[1.]
      ${\Delta^{y,z_2} \vdash U : \ast}$ \hfill{by Lemma~\ref{exampleslem}.}
     \item[2.]
      ${\Delta^{y,z_2,u} \vdash u : U}$ \hfill{by Lemma~\ref{exampleslem}.}
     \item[3.]
     $ {\Delta^{y,z_2} \vdash (\lbind{\delta^u}{u}) : (\pbind{\delta^u}{U}) : \ast}$ \hfill{by $(\la)$ and (app).}
     \item[4.]
       ${\delta^y       \vdash W' : W : \ast}$ \hfill{by $(\la)$ and (app).}
         \end{itemize}
  We have a polymorphic identity function $W'$, but the type $W=\pbind{\delta^{z_2}}{(\pbind{\delta^u}{U})}=\pbind{\delta^{z_2}}{(U\mbto U)}$ doesn't look
  like an intersection type.
  Let's see if we can reach something that looks more like an intersection type instead.
   \vspace{-0.1in}
  \begin{itemize}
   \item[5.]
      ${\Delta^{y,z_2}\mbox{ is legal}}$
 \item[6.]
      ${\Delta^{y,z_2} \vdash y:\ast}$ and 
           ${\Delta^{y,z_2} \vdash \overline{y}:\ast}$ and
      $      {\Delta^{y,z_2} \vdash \overline{y}^2:\ast}$ \hfill{by Lemma~\ref{exampleslem}.}
          \item[7.]
     $      {\Delta^{y,z_2} \vdash z_2:\ast_2}$ and 
           ${\Delta^{y,z_2} \vdash V : \ast}$ \hfill{by Lemma~\ref{exampleslem}.}
         \item[8.]
         $ {\binder{\forall}{i\mbin1..2}
         \:
              \subst{(\pbind{\delta^u}{U})}{z_2}{P_{i,2}}
            =
                    ({P_{i,2}}\,{y}\,{\overline{y}})
              \mbto ({P_{i,2}}\,{y}\,{\overline{y}})
            =_\beta
              \overline{y}^2
            =_\beta}$\\
              $\phantom\quad\quad\quad\quad\quad\quad\quad\quad\quad\quad\quad\quad\quad\quad\quad{({P_{i,2}}\,{\overline{y}}\,{\overline{y}^2})
            =
            \subst{V}{z_2}{P_{i,2}}}$
                   \item[9.]
        $      {\bind{\forall}{i\mbin1..2}
                     {\:\emptyContext \Vdash (\pbind{\delta^u}{U})[z_2:=P_{i,2}]\,\ion \{V[z_2:=P_{i,2}]\}}}$ \hfill \mbox{by (ref) of Figure~\ref{fig:restrsat}.}
          \item[10.]
           ${\gamma^{z_2} \Vdash \pbind{\delta^u}{U}\,\ion\{V\}}$ \hfill \mbox{by {(ctR)} of Figure~\ref{fig:restrsat}.}
      \item[11.]
     $ {\Delta^{y,z_2} \vdash (\lbind{\delta^u}{u}) : V : \ast}$ \hfill{by 3., 10., and (conv).}
      \item[12.]
        ${\delta^y \vdash W' : (\pbind{\delta^{z_2}}{V}) : \ast}$  \hfill{by  11., and ($\la$).}
 \end{itemize}
   The result type here looks better:
   $\pbind{\delta^{z_2}}{V}=\pbind{\delta^{z_2}}{(\appFive{z_2}{}{\overline{y}}{}{\overline{y}^2})}=\mybigcapdot\{{\overline{y}},{\overline{y}^2}\}$.
  It is more obvious that one can simply choose either $\overline{y}$ or $\overline{y}^2$, just like with the intersection type
  $\overline{y}\cap\overline{y}^2$.
  So what would that choice look like?
  Instantiating $\pbind{\delta^{z_2}}{V}$ to either $\overline{y}$ or $\overline{y}^2$ goes like this:

\begin{itemize}
  \item[13.]
    $      {\delta^{y}\mbox{ is legal (as a context)}}$
    \item[14.]
      ${\bind{\forall}{i\mbin1..2}
         {\:\delta^y \vdash P_{i,2} : \ast_2 :\Box}}$ \hfill{by Lemma~\ref{exampleslem} and (weak).}
          \item[15.]
     $      {\bind{\forall}{i\mbin1..2}
         {\:\varepsilon \Vdash P_{i,2}\,\ion\{P_{1,2},P_{2,2}\}}}$ \hfill \mbox{by (ref) of Figure~\ref{fig:restrsat}.}
          \item[16.]
     $      {\bind{\forall}{i\mbin1..2}
         {\:\rdecname(\delta^y) \Vdash P_{i,2}\,\ion\{P_{1,2},P_{2,2}\}}}$ \hfill \mbox{by (ctR) of Figure~\ref{fig:restrsat}.}
         \item[17.]
     $      {\bind{\forall}{i\mbin1..2}
         {\: \delta^y \vdash (\lbind{\delta^u}{u})[z_2 := P_{i,2}] : V[z_2 :=P_{i,2}] : \ast}}$ \hfill{by 11., 15., substitution.}
          \item[18.]
      $      {\binder{\forall}{i\mbin1..2}
         \:
           \subst{(\lbind{\delta^u}{u})}{z_2}{P_{i,2}}
                  = (\lbind{\odec{u}
                        {({P_{i,2}}\,{y}\,{\overline{y}})}}
                  {u})
         \rightarrow_\be
           (\lbind{\odec{u}{\overline{y}^{i-1}}}{u})}$
       \item[19.]  
      ${\binder{\forall}{i\mbin1..2}
         \:
           \subst{V}{z_2}{P_{i,2}}
         = ({P_{i,2}}\,{\overline{y}}\,{\overline{y}^2})
         \rightarrow_\be
           \overline{y}^i}$
         \item[20.]
      $      {\binder{\forall}{i\mbin1..2}
         {\: \delta^y \vdash (\lbind{\odec{u}{\overline{y}^{i-1}}}{u}) : {\overline{y}^i} : \ast}}$ \hfill{by 17., 18., 19., and subject reduction.}
   \end{itemize}
  So both 4., and 12., above are roughly like $(\lbind{u}{u}):\overline{y}\cap\overline{y}^2$ with intersection types and the instantiations are
  like $(\lbind{u}{u}):\overline{y}$ and $(\lbind{u}{u}):\overline{y}^2$.
\end{example}

\begin{example}[Derivation for $\lapp{(\lbind{w}{\lapp{w}{w}})}{(\lbind{u}{u})}$ in Intersection Types Style] \label{derexamplelyint}
  \noindent
  We make use of 1., 4., 14., and 15. of example~\ref{derexamplely} in the following.
  \begin{itemize}
   \item[21.]           ${\delta^{y} \vdash W : \ast : \Box}$
     \hfill{by 1., of example~\ref{derexamplely}.}
     \item[22.]
      ${\Delta^{y,w} \vdash w : W : \ast}$
     \hfill{by 21., \&~(start).}
      \item[23.]
          ${\Delta^{y,w}\mbox{ is legal}}$
      \hfill{by 22., \& definition~\ref{defsystems}.}
   \item[24.]
      ${\bind{\forall}{i\mbin1..2}
       {\:\Delta^{y,w} \vdash P_{i,2} : \ast_2 :\Box}}$
     \hfill{by 23.,  \& 14., of example~\ref{derexamplely}}
   \item[25.]
      ${\bind{\forall}{i\mbin1..2}
        {\:\rdecname(\Delta^{y,w}) \Vdash P_{i,2}\,\ion\{P_{1,2},P_{2,2}\}}}$
      \hfill{by 16., of example~\ref{derexamplely}.}
 \item[26.]
      ${\bind\forall{i\in 1..2}
         {\:
            \Delta^{y,w}
          \vdash
            \newapp{w}{P_{i,2}}
          :
            (\Pi\delta^u.U)[z_2 := P_{i,2}]
          :
            \ast
          }}$
      \hfill{by 22., 24., 25., \&~(app).}
     \item[27.]
      ${\binder\forall{i\in 1..2}\:
          \Delta^{y,w}
        \vdash
          \newapp{w}{P_{i,2}}
        : (\pbind{\odec{u}
                       {({P_{i,2}}\,{y}\,{\overline{y}})}}
                 {({P_{i,2}}\,{y}\,{\overline{y}})})
        : \ast
      }$
       \hfill{by 26.}
       \item[28.]
      ${\Delta^{y,w} \vdash \newapp{w}{P_{1,2}} : \overline{y} : \ast}$
      \hfill{by 27.}
    \item[29.]
      ${
        \Delta^{y,w} \vdash \newapp{w}{P_{2,2}} : {\overline{y}}\mbto{\overline{y}} : \ast}$
     \hfill{by 27.}
\item[30.]
      ${
          \Delta^{y,w}
        \vdash
          (\lapp{\newapp{w}{P_{2,2}}}
                {(\newapp{w}{P_{1,2}})})
        : \overline{y}
        : \ast}$
    \hfill{by 28., 29., \&~(app)}
    \item[31.]
      ${
          \delta^y
        \vdash
          (\lbind{\delta^{w}}
                 {\lapp{\newapp{w}{P_{2,2}}}
                       {(\newapp{w}{P_{1,2}})}})
        : W\mbto\overline{y}
        : \ast}$
     \hfill{by 30., \&~$(\la)$}
    \item[32.]
      ${
         \delta^y
       \vdash
         \lapp{(\lbind{\delta^{w}}
                      {\lapp{\newapp{w}{P_{2,2}}}
                            {(\newapp{w}{P_{1,2}})}})}
              {W'}
       : \overline{y} : \ast}$
     \hfill{by 4., of example~\ref{derexamplely}, 31.}
        \end{itemize}
  This is the equivalent of typing $\lapp{(\lbind{w}{\lapp ww})}{(\lbind{u}{u})}$ with intersection types.
\end{example}

\subsection{Typing Urzyczyn's Untypable Term}
\label{sec:typing-urzyczyn}

Urzyczyn~\cite{Urzyczyn:MSCS-1997} proved $\mycrown{U}=(\la r.\,h(r(\la f \la s.\,f\,s))(r(\la q.\la g.\,g\,q)))(\la o.\,o\,o\,o)$ is
untypable in \textsf{$F_\omega$}.
\cite{GHRDR1993} proved every pure $\la$-term is typable in \textsf{$F_\omega$} iff it is typable in the $\lambda$-cube. Hence $\mycrown{U}$ is untypable in the $\lambda$-cube.
This section types $\mycrown{U}$ in the $\n$-cube by using finite-set declarations.

\begin{definition}[Terms of Type $\ast$ and Sort $\Box$ for Urzyczyn's Term]
  \label{fseqetcdef}
  $$
    \begin{array}[b]{ll}
        F = {z_3}\,{\overline{y}^3}\,{\overline{y}^2}\,{\overline{y}} \hspace{0.8in}
      & Q = {z_3}\,{\underline{y}}\,{y}\,{\underline{y}}
    \\[1ex]
        S = {z_3}\,{\overline{y}^2}\,{\overline{y}^1}\,{y}
      & G = {z_3}\,{\underline{y}_2}\,{\overline{y}}\,{\overline{(\underline{y})}}
    \\[1ex]
      & M = {z_3}\,{\widetilde{(\underline{y})}}\,{y}\,{\underline{y}}
    \\[0.5ex]
        B =  F \mbto S \mbto S
      & A =  Q \mbto G \mbto M
    \\[0.5ex]
        E_1 = \overline{y}^4
      & D_1 = \underline{y}\mbto\underline{y}_2\mbto\widetilde{(\underline{y})}
    \\[0.5ex]
        E_2 = \overline{y}^3
      & D_2 = \underline{y}
    \\[0.5ex]
        E_3 = \overline{y}^2
      & D_3 = \underline{y}_2
    \\[1ex]
        E = \pbind{\delta^{z_3}}{B}
      & D =   \Pi\delta^{z_3}. A
    \\[0.5ex]
        C_1 = \overline{y}^2
      &
        R' = O \mbto C
    \\[0.5ex]
        C_2 = \widetilde{(\underline{y})}
      &
        R_1 = E \mbto C_1
    \\[0.5ex]
        C = {z_2}\,{C_1}\,{C_2}
      &
        R_2 = D \mbto C_2
    \\[0.5ex]
        O = {z_2}\,{E}\,{D}
      &
        R =  \Pi\delta^{z_2}. R'
    \end{array}
  $$
\end{definition}
The proof of the following lemma is straightforward from the typing rules.
\begin{lemma}
\label{typinguntypablelem}
Let $H \in \{F, S, B, Q, G, M, A\}$ and $J \in\{E, D,R\}$.  The following hold:
$\delta^y,\delta^{z_3}\vdash H:\ast$ and $\delta^y\vdash J:\ast$ and $\delta^y,\delta^{z_2}\vdash C:\ast$ and $\delta^y,\delta^{z_2}\vdash R':\ast$.
\end{lemma}

\begin{example}[Viewing $E$, $D$, and $R$ as Intersection Types]
\label{viewingedrasinter}
\noindent
$$
  \begin{array}{|llllllllll|}
    \hline
     E                    & = & \Pi\delta^{z_3}. & F             & \mbto & S                          & \mbto & S                           &                   &
  \\
     \nf(B[z_3 :=P_{1,3}]) & = &                  &\overline{y}^3 & \mbto & \overline{y}^2             & \mbto & \overline{y}^2              & = \overline{y}^4  & = E_1
  \\
     \nf(B[z_3 :=P_{2,3}]) & = &                  &\overline{y}^2 & \mbto & \overline{y}^1             & \mbto & \overline{y}^1              & = \overline{y}^3  & = E_2
  \\
     \nf(B[z_3 :=P_{3,3}]) & = &                  &\overline{y}   & \mbto & y                          & \mbto & y                           & = \overline{y}^2  & = E_3
  \\
  \hline
                          &   &                  &               &                  &                            &                  &                             &                   &
  \\[-2ex]
     D                    & = & \Pi\delta^{z_3}. & Q             & \mbto & G                          & \mbto & M                           &                   &
  \\
     \nf(A[z_3 :=P_{1,3}]) & = &                  &\underline{y}  & \mbto & \underline{y}_2            & \mbto & \widetilde{(\underline{y})} &                   & = D_1
  \\
     \nf(A[z_3 :=P_{2,3}]) & = &                  &{y}            & \mbto & \overline{y}               & \mbto & y                           & = \underline{y}   & = D_2
  \\
     \nf(A[z_3 :=P_{3,3}]) & = &                  &\underline{y}  & \mbto & \overline{(\underline{y})} & \mbto & \underline{y}               & = \underline{y}_2 & = D_3
  \\[1ex]
  \hline
                          &   &                  &               &                  &                            &                  &                             &                   &
  \\[-2ex]
     R                    & = &  \Pi\delta^{z_2}.& O             & \mbto & C                          &                  &                             &                   &
  \\
     \nf(R'[z_2 :=P_{1,2}])& = &                  & E             & \mbto & C_1                        &                  &                             &                   & = R_1
  \\
     \nf(R'[z_2 :=P_{2,2}])& = &                  & D             & \mbto & C_2                        &                  &                             &                   & = R_2
  \\[1ex]
  \hline
  \end{array}
$$
Consider the type $E$ and its component types $F$ and $S$ which we list in separate columns in the table above.  Both $F$ and $S$ act like
tuples of 3 types.  The restriction in the declaration of $z_3$ forces whatever replaces $z_3$ to simply pick one of the three types.  By
looking at the above table, we see that the column with $F$ at the top lists the components of $F$ in the three rows below, and the columns
with $S$ at the top work similarly.  The first row lists $E$ and the three rows below list the results of the three possible instantiations
of $E$.
In effect, $E$ works like the intersection type
\(
    \overline{y}^4 \cap \overline{y}^3 \cap \overline{y}^2
  = E_1            \cap E_2            \cap E_3
\).

The same argument shows that the type $D$ works like the intersection type $D_1\cap{D_2}\cap{D_3}$ and that the type $R$ works like the
intersection type $R_1\cap{R_2}$.
\end{example}

\begin{definition}[Declarations for Urzyczyn's Term]
 Let  $
        \delta^h
      = \odec{h}{(C_1\mbto C_2\mbto y)}$;
       $ \delta^r
      = \odec{r}{R};
    \quad
        \delta^o
      = \odec{o}{O};
    \quad
        \delta^f
      = \odec{f}{F};
    \quad
        \delta^s
      = \odec{s}{S};
    \quad
        \delta^q
      = \odec{q}{Q};
    \quad$ and 
        $\delta^g
      = \odec{g}{G}.
  $%
 \end{definition}

\begin{definition}[Terms of Sort $\ast$ for Urzyczyn's Term] Let \\
  $
    \begin{array}[b]{@{}l@{}}
        T
      =
        \lbind{\delta^{z_3}}{\lbind{\delta^f}{\lbind{\delta^s}{\lapp{f}{s}}}}
    \ \qquad
        J
      =
        \lbind{\delta^{z_3}}{\lbind{\delta^q}{\lbind{\delta^g}{\lapp{g}{q}}}}
    \ \qquad
        L
      =
        \lappThree
          {h}
          {(\lappFour
              {r}
              {}{P_{1,2}}
              {T})}
          {(\lappFour
              {r}
              {}{P_{2,2}}
              {J})}
    \\
        V
      =
        \lbind{\delta^{z_2}}
          {\lbind{\delta^o}
            {\lappFive
               {o}
               {}{P_{1,3}}
               {(\app{o}{}{P_{2,3}})}
               {(\app{o}{}{P_{3,3}})}}}
    \quad\qquad
        U
      =
        \lapp{(\lbind{\delta^r}{L})}{V}
    \end{array}
  $
\end{definition}

Again the following lemma is straightforward according to the typing rules.
\begin{lemma}
\label{typeurz}
The following hold:
\begin{enumerate}
\item
\label{typeurzone}
\begin{enumerate}
\item $\gamma^{z_3} \Vdash F\,\ion\{S\mbto S\}$
\item $\gamma^{z_3} \Vdash G\,\ion\{Q\mbto M\}$,
\item
  $\gamma^{z_2}
  \Vdash O\,\ion\{\pbind{\delta^{z_3}}{(\appFive{z_2}{}{B}{}{A})}\}$%
\end{enumerate}
\item
\label{typeurztwo}
$\delta^y\vdash T :E:\ast$.
\item
\label{typeurzthree}
$\delta^y\vdash J :D:\ast$.
\item
  \label{typeurz4}
  \begin{enumerate}
  \item
    \label{typeurz4a}
    \(
        \Delta^{y,{z_2},o}
      \vdash
        (\app{o}{}{P_{i,3}})
      : ({
        \appFive
          {z_2}
          {}{E_i}
          {}{D_i}})
      : \ast
    \)
    \relax
    for $i\in\{1,2,3\}$.
  \item
    \label{typeurz4d}
    $\Delta^{y, {z_2}, o}\vdash (o\, P_{1,3})(o\, P_{2,3})(o\, P_{3,3}):C:\ast$.
  \item
    \label{typeurz4e}
    $\delta^y\vdash V : R :\ast$.
  \end{enumerate}
\item
  \label{typeurz5}
  \begin{enumerate}
  \item
    \label{typeurz5a}
    $\Delta^{y,h,r} \vdash (\app{r}{}{P_{i,2}}) : R_i : \ast$ for $i\in\{1,2\}$.
  \item
    \label{typeurz5c}
    $\Delta^{y,h,r} \vdash (\appFive{r}{}{P_{1,2}}{}{T}) : C_1 : \ast$.
  \item
    \label{typeurz5d}
    $\Delta^{y,h,r} \vdash (\appFive{r}{}{P_{2,2}}{}{J}) : C_2 : \ast$.
  \end{enumerate}
\item
\label{typeurz6}
$\Delta^{y,h,r} \vdash L:y:\ast$ and $\Delta^{y,h} \vdash \la\delta^r. L : \Pi \delta^r. y :\ast$.
\item
\label{typeurz7}
$\Delta^{y,h} \vdash U : y : \ast$.
\end{enumerate}
\end{lemma}
Now we show that Urzyczyn's famous term is typable in the $\n$-cube.
\begin{example}[Urzyczyn's Term Is Typable]
  Clearly, Urzyczyn's term $\mycrown{U}=\TE(U)$.
  Since Lemma~\ref{typeurz}.\ref{typeurz7} shows that
  $\Delta^{y,h} \vdash U : y : \ast$, then by Definition~\ref{typabilitydef}, $\mycrown{U}$ is typable.
\end{example}

\section{Conclusion}
In this paper we introduced an extension of the PTS $\la$-cube using finite set
declarations that allow us to translate intersection types as $\la$-terms. We gave
the translation of Urzyczyn's famous term $U$ (which is untypable in the $\la$-cube)
in the $\n$-cube and showed that this term is indeed typable in the $\n$-cube. The
set up and machinery presented in this paper can be followed to prove that the
$\n$-cube characterizes all strongly normalising terms.

\end{document}